\documentclass{article}[12pt]
\usepackage{amsmath,amssymb,amsthm}

\usepackage[backend=bibtex,firstinits=true,style=trad-plain]{biblatex}
\usepackage{tikz}

\tikzstyle{normalNode}=[draw=black, fill=black, circle, minimum size=1.35em, inner sep=0.5mm]
\tikzstyle{labeledNode}=[normalNode, fill=white, minimum size=1.35em, inner sep=0.5mm]
\tikzstyle{normalEdge}=[black, thick, >=stealth]
\tikzstyle{secondaryEdge}=[normalEdge, dashed]

\usepackage[colorlinks,linkcolor=blue,citecolor=blue,urlcolor=blue]{hyperref}
\usepackage[capitalise]{cleveref}

\newtheorem{theorem}{Theorem}
\newtheorem{lemma}[theorem]{Lemma}
\newtheorem{corollary}[theorem]{Corollary}

\def\mathrlap{\mathpalette\mathrlapinternal} 
\def\mathrlapinternal#1#2{\rlap{$\mathsurround=0pt#1{#2}$}}

\newcommand{\elsum}[1]{\sum_{\mathrlap{#1}}\;} 

\newcommand\st{\;:\;}

\newcommand{\RFprimal}{\text{[P]}}
\newcommand{\RFdual}{\text{[D]}}

\newcommand{\rflow}{\textsc{Maximum Robust Flow}}
\newcommand{\irflow}{\textsc{Integral Maximum Robust Flow}}
\newcommand{\robustval}{\operatorname{val_r}}
\newcommand{\flowval}{\operatorname{val}}

\newcommand{\supp}[1]{\operatorname{supp}(#1)}

\newcommand{\NP}{N\!P}

\title{The Complexity of Computing a Robust Flow}

\author{Yann Disser\thanks{Department of Mathematics, Graduate School CE, TU Darmstadt} \and Jannik Matuschke\thanks{Department of Mathematics and School of Management, Technische Universit\"at M\"unchen}} 

\date{}

\begin{document}
\maketitle

\begin{abstract}
Robust network flows are a concept for dealing with uncertainty and unforeseen failures in the network infrastructure. 
  One of the most basic models is the {\rflow} problem: Given a network and an integer $k$, the task is to find a path flow of maximum robust value, i.e., the guaranteed value of surviving flow after removal of any $k$ arcs in the network. The complexity of this problem appeared to have been settled a decade ago: Aneja et al.~\cite{AnejaChandrasekaranNair2001} showed that the problem can be solved efficiently when $k = 1$, while an article by Du and Chandrasekaran~\cite{DuChandrasekaran2007} established that the problem is $\NP$-hard for any constant value of~$k$~larger~than~$1$.
  
  We point to a flaw in the proof of the latter result, leaving the complexity for constant $k$ open once again.
  For the case that $k$ is not bounded by a constant, we present a new hardness proof, establishing $\NP$-hardness even for instances where the number of paths is polynomial in the size of the network.
  We further show that computing optimal integral solutions is already $\NP$-hard for $k = 2$ (whereas for $k=1$, an efficient algorithm is known) and give a positive result for the case that capacities are in $\{1, 2\}$. 
\end{abstract}

\section{Introduction}
Network flows are an important tool for modeling vital network services, such as transportation, communication, or energy transmission.
In many of these applications, the flow is subjected to uncertainties such as failures of links in the network infrastructure. 
This motivates the study of \emph {robust optimization} versions of network flows, which offer a concept to anticipate and counteract such failures. 
A fundamental optimization problem within this framework is to find a flow that maximizes the amount of surviving flow after it is affected by a worst-case failure of $k$ links in the network for some fixed number $k$. This problem is also known as the {\rflow} problem. 

In this paper, we discuss the complexity of {\rflow}. We point out an error in an earlier result on this problem, which claimed $\NP$-hardness for the case that $k$ is fixed to a constant value larger than~$1$. In its stead, we give a new hardness proof which, however, requires the number $k$ to be a non-constant part of the input. 
Before we discuss these results in detail, we give a formal definition of the problem and discuss related literature.

\paragraph{Problem definition}
We are given a directed graph $G = (V, E)$ with source $s$, sink $t$, capacities $u \in \mathbb{Z}_{+}^E$, and an integer $k$, specifying the number of possible link failures.
Let $\mathcal{P}$ denote the set of $s$-$t$-paths in $G$ and let $\mathcal{S} := \{S \subseteq E \st |S| = k\}$.
  An \emph{$s$-$t$-flow} is a vector $x \in \mathbb{R}_+^\mathcal{P}$ respecting the capacity constraints $\sum_{P : e \in P} x(P) \leq u(e)$ for all $e \in E$. 
  The goal is to find an $s$-$t$-flow $x$ that maximizes the \emph{robust flow value}  
  $$\robustval(x) := \sum_{P \in \mathcal{P}} x(P) - \max_{S \in \mathcal{S}} \  \elsum{\ \ P \in \mathcal{P} : P \cap S \neq \emptyset} x(P),$$
 i.e., the amount of remaining flow after failure of any set of $k$ arcs.
 
 \paragraph{Related work}
 Aneja, Chandrasekaran, and Nair~\cite{AnejaChandrasekaranNair2001} were the first to investigate {\rflow}. They showed that if $k = 1$, the problem can be solved in polynomial time by solving a parametric linear program. In fact, their LP yields a flow $x$ that simultaneously maximizes $\robustval(x)$ and the nominal flow value $\flowval(x) := \sum_{P \in \mathcal{P}} x(P)$.
They also show that a maximum \emph{integral} robust flow can be found in polynomial time for $k=1$, even though its value might be strictly lower than that of the optimal fractional solution.
 Following up on this work, Du and Chandrasekaran~\cite{DuChandrasekaran2007} investigated the problem for values of $k$ larger than~$1$. They presented a hardness proof for {\rflow} with $k=2$. Unfortunately, however, this proof is incorrect. We explain this error in detail in \cref{sec:old-hardness}.

Because of the presumed hardness of the problem, later work focused on approximation algorithms.
Bertsimas, Nasrabadi, and Stiller~\cite{BertsimasNasrabadiStiller2013} use a variation of the parametric LP to obtain an approximation algorithm for {\rflow} whose factor depends on the fraction of flow lost through the failure. More recently, Bertsimas, Nasrabadi, and Orlin~\cite{bertsimas2013power} gave an alternative analysis of the same algorithm, establishing an approximation factor of $1 + (k/2)^2/(k+1)$.
Another related concept are \emph{$k$-route flows} introduced by Aggarwal and Orlin~\cite{aggarwal2002multiroute}. A $k$-route flow is a conic combination of elementary flows, each sending flow uniformly along $k$ disjoint paths. This structure ensures that the failure of any arc can only destroy a $1/k$ fraction of the total flow. Baffier et al.~\cite{baffier2016parametric} observed that computing a maximum $(k+1)$-route flow yields a $(k+1)$-approximation for {\rflow}.

Several alternative robustness models for flows have been proposed in different application contexts. 
Taking a less conservative approach, Bertsimas, Nasrabadi, and Stiller~\cite{BertsimasNasrabadiStiller2013} and Matuschke, McCormick, and Oriolo~\cite{matuschke2017rerouting} proposed different models of flows that can be rerouted after failures occur.
Matuschke et al.~\cite{matuschke2016protecting} investigated variants of robust flows in which an adversary can target individual flow paths and the network can be fortified against such attacks.
Gottschalk et al.~\cite{gottschalk2016robust} devised a robust variant of flows over time in which transit times are uncertain.

Robust flows can be seen as a dual version of \emph{network flow interdiction}, where the task is to find a subset of $k$ arcs whose removal minimizes the maximum flow value in the remaining network. Wood~\cite{wood1993deterministic} proved that this problem is strongly $\NP$-hard. The reduction presented in \cref{sec:new-hardness} also exploits the fact that interdiction is $\NP$-hard, but the construction is considerably more involved in order to couple network flow and interdiction decisions in the correct way. For an overview of results on network flow interdiction, see the recent article by Chestnut and Zenklusen~\cite{chestnut2017hardness} on the approximability of the problem.

\paragraph{Results and structure of this paper}
In \cref{sec:background}, we give the background necessary to understand Du and Chandrasekaran's reduction~\cite{DuChandrasekaran2007} and the reason why it does not imply hardness for {\rflow} with $k=2$.

In \cref{sec:new-hardness}, we then give a new reduction that establishes $\NP$-hardness for {\rflow} when $k$ is an arbitrarily large number given in the input.
Our reduction even works when  the number of paths in the graph is polynomial in the size of the network and only two different capacity values occur (capacity $1$ and a capacity that is large but polynomial in the size of the network). We also point out that the problem becomes easy for the case that all capacities are equal.

In \cref{sec:integral}, we show that it is $\NP$-hard to compute an optimal \emph{integral} solution for $k=2$. Note that this is in contrast to the case $k=1$, where the optimal integral solution can be computed efficiently~\cite{AnejaChandrasekaranNair2001}. While  $\NP$-hardness for the integral case even holds when capacities are bounded by $3$, we show that the problem can be solved efficiently when capacities are bounded by $2$, even for arbitrary values of~$k$.

 \section{Background}\label{sec:background}

 The hardness result in~\cite{DuChandrasekaran2007} is based on an LP formulation of {\rflow} and the equivalence of optimization and separation, which we shortly recapitulate in this section.
 
 \subsection{LP formulation}
 For our further discussion of {\rflow}, the following linear programming formulation of the problem will be useful:
 
 \begin{alignat*}{3}
	 \RFprimal\quad \max && \elsum{P \in \mathcal{P}} x(P) & - \lambda\\
	 \text{s.t.} && \elsum{P : e\in P} x(P) & \, \leq \, u(e) && \qquad \forall\; e \in E\\
   && \elsum{P : P \cap S \neq \emptyset} x(P) - \lambda & \, \leq \, 0 && \qquad \forall\; S \in \mathcal{S}\\
   && x(P) & \, \geq \, 0 && \qquad \forall\; P \in \mathcal{P}
 \end{alignat*}
 
 Note that $\lambda = \max_{S \in \mathcal{S}} \sum_{P \in \mathcal{P} : P \cap S \neq \emptyset} x(P)$ in any optimal solution to $\RFprimal$,~i.e.,~$\lambda$ represents the amount of flow lost in a worst-case failure scenario for flow~$x$. We also consider the dual of $\RFprimal$:
  \begin{alignat*}{3}
	 \RFdual\quad \min && \sum_{e \in E} u(e) y(e) \\
	 \text{s.t.} && \ \ \sum_{e\in P} y(e) \ + \ \elsum{S : P \cap S \neq \emptyset} z(S) & \,\geq\, 1 && \qquad \forall\; P \in \mathcal{P}\\
   && \sum_{S \in \mathcal{S}} z(S) & \,=\, 1\\
   && y(e) & \, \geq \, 0 && \qquad \forall\; e \in E\\
   && z(S) & \, \geq \, 0 && \qquad \forall\; S \in \mathcal{S}
 \end{alignat*}
 
   Note that the number of $s$-$t$-paths in $G$ and hence the number of variables of $\RFprimal$ can be exponential in $|E|$. On the other hand, the number of variables of $\RFdual$ is $|E| + \binom{|E|}{k}$, which is polynomial in $|E|$ for constant values of $k$. In such a situation, a standard approach is to solve the dual via its \emph{separation problem}, which is described in the next section.
   
 \subsection{Equivalence of Optimization and Separation}
 Let $Q \subseteq \mathbb{R}^n$ be a rational polyhedron. 
 By a classic result of Gr\"otschel, Lovasz, and Schrijver~\cite{groetschel1988geometric}, optimizing arbitrary linear objectives over $Q$ is polynomially equivalent to finding out wether a given point is in $Q$ and finding a hyperplane separating the point from $Q$ if not. We give a formal statement of this result below.

\medskip
\noindent
$\textsc{Separation}(Q)$\\[-18pt]
\begin{description}
\item[Input:] a vector $y \in \mathbb{R}^n$\\[-18pt]
\item[Task:] Assert that $y \in Q$, or find a separating hyperplane, i.e., a vector $d \in \mathbb{R}^n$ such that $d^{T}x < d^{T}y$ for all $x \in Q$.
\end{description} 
 
\noindent
 $\textsc{Optimization(Q)}$\\[-18pt]
 \begin{description}
\item[Input:] a vector $c \in \mathbb{R}^n$\\[-18pt]
\item[Task:] Either assert that $Q = \emptyset$, or find $x, d \in \mathbb{R}^n$ such that $c^Td > 0$ and $x + \alpha d \in Q$ for all $\alpha \geq 0$, or find $x \in Q$ maximizing $c^Tx$.
 \end{description}
 
 \begin{theorem}[{\cite[Theorem~6.4.9]{groetschel1988geometric}}]
   The optimization problem for $Q$ can be solved in oracle-polynomial time given an oracle for the separation problem for~$Q$, and vice versa.
 \end{theorem}
 
 \subsection{Dual Separation for {\rflow}}\label{sec:old-hardness}
 
 Let $Q$ be the set of feasible solutions of the dual program $\RFdual$, i.e., $$Q := \left\{(y, z) \in \mathbb{R}^{E \times \mathcal{S}} \st \sum_{S \in \mathcal{S}} z(S) = 1,\ \sum_{e \in P} y(e) + \elsum{S : P \cap S \neq \emptyset} z(S) \geq 1 \ \forall P \in \mathcal{P}\right\}.$$
 
 In the separation problem for $Q$, we are given $(y, z) \in \mathbb{R}^{E \times \mathcal{S}}$ and have to decide whether $(y, z) \in Q$. Since checking whether $\sum_{S \in \mathcal{S}} z(S) = 1$ can be done in polynomial time for constant values of $k$, the separation problem is polynomial-time equivalent to finding a path $P$ such that $\sum_{e \in P} y(e) + \sum_{S : P \cap S \neq \emptyset} z(S) < 1$ or deciding that no such path is exists.
 
Du and Chandrasekaran~\cite{DuChandrasekaran2007} showed that $\textsc{Separation}(Q)$ is $\NP$-hard, even when $k = 2$. They concluded that by the equivalence of optimization and separation, solving $\RFdual$ and hence solving $\RFprimal$ is NP-hard. However, this claim is not correct. It is true that the hardness of $\textsc{Separation}(Q)$ implies that also $\textsc{Optimization}(Q)$ is $\NP$-hard. However, $\RFdual$ is only a special case of $\textsc{Optimization}(Q)$: The objective function of $\RFdual$ is not an arbitrary vector in $\mathbb{R}^{E \times \mathcal{S}}$, but it is restricted to those objective functions where all coefficients corresponding to the $z$-variables are $0$. Indeed, it turns out that the instances of {\rflow} with $k = 2$ constructed in the reducton of~\cite{DuChandrasekaran2007} contain an $s$-$t$-cut of cardinality $2$---implying every $s$-$t$-flow has robust value $0$ in these instances.

\section{Robust flows with large number of failing arcs}\label{sec:new-hardness}

\begin{theorem}\label{thm:hardness-arbitrary-k}
  {\rflow} is strongly NP-hard, even when restricted to instances where the number of paths is polynomial in the size of the graph.
\end{theorem}

\begin{proof}
We show this by a reduction from \textsc{Clique}: Given a graph $G' = (V', E')$ and $k' \in \mathbb{Z}_+$, is there a clique of size $k'$ in $G'$?
We will construct an instance of {\rflow} consisting of a graph $G = (V, E)$, source $s$, sink $t$, capacities $u\colon E \rightarrow \mathbb{Q}_{+} \cup \{\infty\}$, and $k \in \mathbb{Z}_+$ from the \textsc{Clique} instance (at the end of the proof, we show how to obtain an equivalent instance with finite and integral capacities).
Let $$\ell := |V'| + 2|E'|,\quad k := k' \ell + (|V'| - k') + 2|E'|,\quad \varepsilon := \frac{1}{\ell},\quad M := (1+\varepsilon)k.$$
For every vertex $v \in V'$ we introduce a node $a_v$ and two additional groups of $\ell$ nodes each, $A_v = \{a_{v, 1}, \dots, a_{v, \ell}\}$ and $B_v = \{b_{v, 1}, \dots, b_{v, \ell}\}$. We connect $a_v$ to every node in $B_v$ by an arc of capacity $M$, and we also connect each node $a_{v, i}$ to $b_{v, i}$ by an arc of capacity $1$.
For every edge $e = \{u, v\} \in E'$ we introduce two nodes $a'_e, a''_e$ and arcs $(a'_e, b_{u, i}), (a''_e, b_{u, i}), (a'_e, b_{v, i}), (a''_e, b_{v, i})$ for $i \in \{1, \dots, \ell\}$, each of capacity $M$.
We denote 
$$A := \bigcup_{v \in V'} (\{a_v\} \cup A_v) \cup \bigcup_{e \in E'} \{a'_e, a''_e\} \text{ and } B := \bigcup_{v \in V'} B_v.$$ 
We also introduce a source $s$ and a sink $t$ and arcs $(s, a)$ for every $a \in A$ and $(b, t)$ for every $b \in B$, all of infinite capacity.
We then add $k$ parallel $s$-$t$-arcs $e_1, \dots, e_{k}$. Defining $h := 2 \cdot \binom{k'}{2} - 2$, we set the capacity of $e_1, \dots, e_h$ to $1 + \varepsilon$ and the capacity of $e_{h+1}, \dots, e_{k}$ to $1$. 
We finally add two additional nodes $v', v''$, together with two $s$-$v'$-arcs $e'_1$, $e'_2$, two $v''$-$t$-arcs $e''_1$, $e''_2$, and arcs $(s, v''), (v', t), (v', v'')$. We set the capacities $u(e'_1) = u(e''_1) = 1$, $u(e'_2) = u(e''_2) = u(v', v'') = \varepsilon$ and $u(s, v'') = u(v', t) = 1 + \varepsilon$. We let $E_H$ denote the arcs in the subgraph $H$ induced by the node set $\{s, v', v'', t\}$.

\begin{figure}[t]
  \centering
	\begin{tikzpicture}[font=\scriptsize]
		\node[labeledNode] (s) {$s$};
		\path (s)
			++(0, 2) node[labeledNode] (v) {$a_v$}
				edge[normalEdge, bend left=20, <-] (s)
			++(0.8, 0) node[labeledNode] (v1) {$a_{v,1}$}
			  edge[normalEdge, bend left=20, <-] (s)
			++(0.8, 0) node (vi) {$\dots$}			
			++(0.8, 0) node[labeledNode] (vl) {$a_{v,\ell}$}
				edge[normalEdge, bend left=20, <-] (s)
			++(1.2, 0) node[labeledNode] (e1) {$a'_e$}
				edge[normalEdge, bend left=20, <-] (s)
			++(0.8, 0) node[labeledNode] (e2) {$a''_e$}
				edge[normalEdge, bend left=20, <-] (s)
			++(1.2, 0) node[labeledNode] (w1) {$a_{w,1}$}
			  edge[normalEdge, bend left=20, <-] (s)
			++(0.8, 0) node (wi) {$\dots$}			
			++(0.8, 0) node[labeledNode] (wl) {$a_{w,\ell}$}
				edge[normalEdge, bend left=20, <-] (s)
			++(0.8, 0) node[labeledNode] (w) {$a_w$}
				edge[normalEdge, bend left=30, <-] (s)
			++(-7.2, 2) node[labeledNode] (v1b) {$b_{v,1}$}
				edge[normalEdge, <-] (v)
				edge[normalEdge, <-, dashed] (v1)
				edge[normalEdge, <-] (e1)
				edge[normalEdge, <-] (e2)
			++(0.8, 0) node (vib) {$\dots$}
			++(0.8, 0) node[labeledNode] (vlb) {$b_{v,\ell}$}
				edge[normalEdge, <-] (v)
				edge[normalEdge, <-, dashed] (vl)
				edge[normalEdge, <-] (e1)
				edge[normalEdge, <-] (e2)
			++(3.2, 0) node[labeledNode] (w1b) {$b_{w,1}$}
				edge[normalEdge, <-] (w)
				edge[normalEdge, <-, dashed] (w1)
				edge[normalEdge, <-] (e1)
				edge[normalEdge, <-] (e2)
			++(0.8, 0) node (wib) {$\dots$}
			++(0.8, 0) node[labeledNode] (wlb) {$b_{w,\ell}$}
				edge[normalEdge, <-] (w)
				edge[normalEdge, <-, dashed] (wl)
				edge[normalEdge, <-] (e1)
				edge[normalEdge, <-] (e2)
			++(-7.2, 2) node[labeledNode] (t) {$t$}
				edge[normalEdge, bend left=20, <-] (v1b)
				edge[normalEdge, bend left=20, <-] (vlb)
				edge[normalEdge, bend left=20, <-] (w1b)
				edge[normalEdge, bend left=20, <-] (wlb)
				edge[normalEdge, bend right=20, <-, dashed] node[right] {$e_k$} node[left] {$\dots$} (s)
				edge[normalEdge, bend right=50, <-, dashed] node[left] {$e_1$} (s);
			
			\path (s) 
			  ++(-2.75, 3) node[labeledNode] (vv) {$v'$}
			    edge[normalEdge, double, <-, bend right=30, dashed] node[above, sloped] {$e'_1$} node[below, sloped] {$e'_2$} (s)
			    edge[normalEdge, ->, bend left=30, dashed] (t)
			  ++(-1.5, 0) node[labeledNode] (vvv) {$v''$}
			    edge[normalEdge, <-, bend right=45, dashed] (s)
			    edge[normalEdge, <-, dashed] (vv)
			    edge[normalEdge, double, ->, bend left=45, dashed] node[above, sloped] {$e''_1$} node[below, sloped] {$e''_2$} (t);
			    
	\end{tikzpicture}%
	\vspace{-0.5cm}
	\caption{Construction of the reduction from \textsc{Clique} for two vertices $v, w$ and an edge $e = \{v, w\}$. Solid arcs have a ``large'' capacity (i.e., $u(e) \in \{M, \infty\}$), dashed arcs have a ``small'' capacity (i.e., $u(e) \in  \{\varepsilon, 1, 1 + \varepsilon\}$).}
\end{figure}
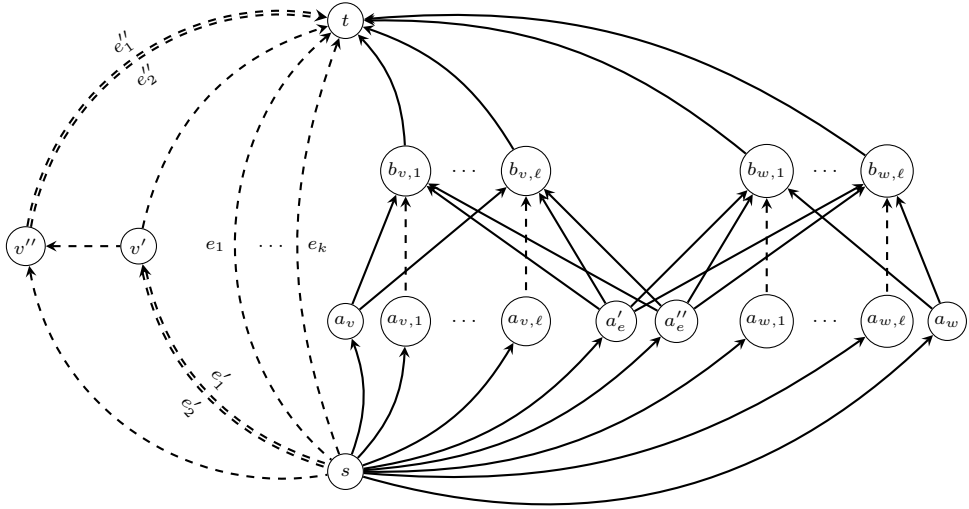

We now prove the following lemma, which implies Theorem~\ref{thm:hardness-arbitrary-k}. For convenience we will use the notation $x(e) := \sum_{P : e \in P} x(P)$ for the total flow through an arc $e$.

\begin{lemma}\label{lem:flow-value-clique}
  Let $(x^*, \lambda^*)$ be an optimal solution to {\rflow}. Then there is a clique of size $k'$ in $G'$ if and only if $x^*(v', v'') > 0$.
\end{lemma}

\begin{proof}
In order to prove Lemma~\ref{lem:flow-value-clique} we first observe that, without loss of generality, we can assume that all arcs in $E \cap (A \times B)$ and the arcs $e_1, \dots, e_{k}$ are saturated by $x^*$: If any of these arcs is not saturated, we can increase the flow along the unique path containing that arc and increase $\lambda^*$ by the same value, not decreasing the value of the solution and not changing the flow on $(v', v'')$. 

Consider the set $F := \{e_{1}, \dots, e_{k}\} \cup E_H$ and define $$f_{x^*}(r) := \max \left\{\sum_{e \in F'} x^*(e) \, :\, F' \subseteq F,\, |F'| \leq r\right\}$$
for $r \in \mathbb{N}$. We derive the following lemma.

\begin{lemma}\label{lem:best-response-value}
  Let $h^* := \max \{|E'[U]| \,:\, U \subseteq V', \; |U| \leq k' \}$. Then,
  $$\lambda^* = (|V'| + 4|E'|)\ell M + k'\ell + f_{x^*}(2h^*).$$
\end{lemma}

\begin{proof}
  Let $S \in \mathcal{S}$ be such that $\sum_{P \in \mathcal{P} : S \cap P \neq \emptyset} x^*(P) = \lambda^*$.
  Without loss of generality, we can assume that $S \cap (A \times B) = \emptyset$: If $S$ contains an arc $(a, b) \in A \times B$, we can replace it by either of the arcs $(s, a)$ or $(b, t)$, each of which intersect the unique $s$-$t$-path containing $(a, b)$. 
  Now define $$U := \{v \in V' \,:\, (b, t) \in S \text{ for all } b \in B_v \}.$$ Note that $|U| \leq \lfloor k/\ell \rfloor \leq k'$ by choice of $k$ and $\ell$. 
  Furthermore, note that $x^*(P) = M$ for exactly $(|V'| + 4|E'|)\ell$ paths $P \in \mathcal{P}$ by our earlier assumption that arcs in $E \cap (A \times B)$ are fully saturated.
Also, by choice of $M$ and since every other path carries at most $1+\varepsilon$ units of flow, the only possibility to destroy at least $(|V'| + 4|E'|)\ell M$ units of flow is for $S$ to intersect all these paths, and by maximality of $\lambda^*$, this must indeed be the case.
Therefore, we can assume that for every $v \in V'$, either $v \in U$ or $\{(s, a_v)\} \cup \{(s, a'_e), (s, a''_e) \,:\, e \in \delta(v)\} \subseteq S$. This implies that $U$ already determines a subset $S_U$ of $$k_U := \ell |U| + |V'| - |U| + 2(|E'| - |E'[U]|)$$ arcs in $S$, destroying a flow of $(|V'| + 4|E'|)\ell M + |U| \ell$ units. The remaining $k - k_U$
  arcs in $S$ can destroy an additional flow of at most $f_{x^*}(k - k_U)$, as no arc in $E \setminus F$ carries more than $1$ unit of flow after destruction of the flow paths of value $M$ and there are at least $k$ arcs in $F$ with flow value at least $1$.
  Furthermore observe that $f_{x^*}(r' + r'') \leq f_{x^*}(r') + (1+\varepsilon)r''$ as none of the arcs in $F$ carries more than $1 + \varepsilon$ units of flow.
  We deduce that
  \begin{align*}
     \lambda^* & \leq \; (|V'| + 4|E'|)\ell M + |U| \ell + f_{x^*}(k - k_U)\\
     & = \; (|V'| + 4|E'|)\ell M + |U| \ell + f_{x^*}((k' - |U|)(\ell - 1) + 2|E'[U]|)\\
     & \leq \; (|V'| + 4|E'|)\ell M + |U|\ell + f_{x^*}(2|E'[U]|) + (1 + \varepsilon)(k' - |U|)(\ell - 1)\\
    & = \; (|V'| + 4|E'|)\ell M + k'\ell + (k' - |U|)(\underbrace{\varepsilon(\ell - 1) - 1}_{\leq 0}) + f_{x^*}(2\underbrace{|E'[U]|}_{\leq h^*})\\
    & \leq (|V'| + 4|E'|)\ell M + k'\ell + f_{x^*}(2h^*).
  \end{align*}
   Now let $U^* \subseteq V'$ be such that $|U^*| = k'$ and $|E'[U^*]| = h^*$ and let $F^* \subseteq F$ be such that $|F^*| = 2h^*$ and $\sum_{e \in F^*} x^*(e) = f_{x^*}(2h^*)$ (recall that we may assume arcs in $F$ to be saturated). Consider the set
   $$S^* := \bigcup_{v \in U^*} \!\! B_v \, \cup \, \{a_v \,:\, v \in V' \setminus U^*\} \, \cup \, \{a'_e, a''_e \,:\, e \notin E'[U^*]\} \, \cup \, F^*$$
  and observe that $\sum_{P : P \cap S^* \neq \emptyset} x^*(P) = (|V'| + 4|E'|)\ell M + k'\ell + f_{x^*}(2h^*)$. This proves Lemma~\ref{lem:best-response-value}.
\end{proof}

 We use Lemma~\ref{lem:best-response-value} to prove Lemma~\ref{lem:flow-value-clique} as follows. Observe that $(x^*, \lambda^*)$ maximizes the quantity $\sum_{P \in \mathcal{P}} x^*(P) - \lambda^*$. As we already fixed the flow value on all paths outside of the subgraph $H$, we know that 
 \begin{align*}
 \sum_{P \in \mathcal{P}} x^*(P) & = \underbrace{(|V'| + 4|E'|)\ell M + \ell|V'| + \sum_{i = 1}^{k} u(e_i)}_{C_1 :=} + \sum_{P \in \mathcal{P} : P \subseteq E_H} x^*(P)\\
 & = C_1 + x^*(v', t) + x^*(v', v'') + x^*(s, v''),
 \end{align*}
 where the last three summands together determine the total nominal flow through $H$. Defining $C_2 := (|V'| + 4|E'|)\ell M + k'\ell$, Lemma~\ref{lem:best-response-value} states that $\lambda^* = C_2 + f_{x^*}(2h^*)$. As $C_1$ and $C_2$ do not depend on the flow in $E_H$, we deduce that the flow $x^*$ in $E_H$ maximizes the quantity $$x^*(v', t) + x^*(v', v'') + x^*(s, v'') - f_{x^*}(2h^*).$$ 
 
 First, assume $G'$ has no clique of size $k'$, i.e., $h^* \leq \binom{k'}{2} - 1$. In this case, $2h^* \leq h$ and hence $f_{x^*}(2h^*) = 2h^*(1+\varepsilon)$, independent of the flow values in the subgraph $H$, as no arc in $E_H$ can carry more than $1 + \varepsilon$ units of flow and there are already $h$ arcs with flow value $1 + \varepsilon$ in $F \setminus E_H$. Therefore, $x^*$ maximizes $x^*(v', t) + x^*(v', v'') + x^*(s, v'')$, which implies it is the unique maximum flow in $H$ fulfilling $\sum_{P : (v', v'') \in P} x^*(P) = 0$.
 
 Now assume $G'$ has a clique of size $k'$ and thus $h^* = \binom{k'}{2}$. In this case $2h^* = h + 2$ and hence $$f_{x^*}(2h^*) = 2h \cdot (1+\varepsilon) + \max\{1,\, x^*(v', t)\} + \max \{1,\, x^*(s, v'')\},$$ 
as $(v', t)$ and $(s, v'')$ are the only two arcs in $F$ outside $\{e_1, \dots, e_h\}$ that can carry more than $1$ unit of flow. Thus $x^*$ maximizes
 $$x^*(v', t) + x^*(v', v'') + x^*(s, v'') - \max\{1,\, x^*(v', t)\} - \max \{1,\, x^*(s, v'')\}.$$ This term is maximized for $x^*(v', t) = x^*(s, v'') = 1$ and $x^*(v', v'') = \varepsilon$.
 
 The above two observations conclude the proof of Lemma~\ref{lem:flow-value-clique}.\end{proof}
 
 Note that the size of the graph $G = (V, E)$ constructed in the reduction is polynomial in the size of $G'$. Furthermore observe that $|\mathcal{P}| \leq |E|$ and that all capacities are polynomial in the size of $G$ (note that the capacity $\infty$ can be replaced by $|E|M$, and multiplying all capacities with $\ell$ yields integral capacities).
  This concludes the proof of Theorem~\ref{thm:hardness-arbitrary-k}.
\end{proof}

\paragraph{Reduction to two capacity values} We now observe that there is a pseudopolynomial transformation that converts general instances of {\rflow} to instances with where the capacity of each arc is one of two different values: $1$ or $u_{\max}$, where $u_{\max}$ is the maximum capacity value occurring in the original instance.

\begin{lemma}\label{lem:two-capacities}
There is an algorithm that given an instance of {\rflow} $I = ((V, E), s, t, u, k)$ computes in time $\mathcal{O}(|E| u_{\max})$ an instance of {\rflow} $I' = ((V', E', s', t', u', k)$ with $u'(e) \in \{1, u_{\max}\}$ for all $e \in E'$ such that the maximum robust flow value of $I$ and $I'$ is identical, where $u_{\max} := \max_{e \in E} u(e)$. Moreover, given an (integral) flow $x'$ in $I'$ one can compute in time polynomial in $|E|$ and $|\supp{x'}|$  an (integral) flow $x$ in $I$ with $\robustval(x) = \robustval(x')$.
\end{lemma}

\begin{figure}[t]
  \centering
	\begin{tikzpicture}[font=\scriptsize]
		\path node[labeledNode] (s) {$v$}
			++(2, 0) node[normalNode] (v) {}
				edge[normalEdge, <-] node[above] {$\infty$} (s)
			++(2, 0) node[labeledNode] (w) {$w$}
			  edge[normalEdge, <-, bend left=70] node[above=-2pt] {$1$} (v)
			  edge[normalEdge, <-, bend left=35] node[above=-2pt] {$1$} (v)
			  edge[normalEdge, <-, bend left=0] node[above=-2pt] {$1$} (v)
			  edge[normalEdge, <-, bend left=-35] node[above=-2pt] {$1$} (v)
			  edge[normalEdge, <-, bend left=-70] node[above=-2pt] {$1$} (v);
	\end{tikzpicture}%
	\vspace{-0.4cm}
	\caption{Construction for Lemma~\ref{lem:two-capacities}. Arc $(v, w)$ with capacity $5$ is replaced by a sequence of an arc with capacity $\infty$ and $5$ arcs with capacity $1$.\label{fig:two-capacities}}
\end{figure}
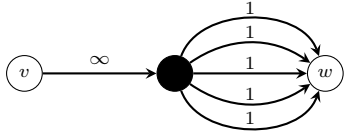

\begin{proof}
To obtain an equivalent instance in which only the capacities $1$ and $u_{\max}$ occur, observe that we can replace any arc of capacity $u$ by a the concatenation of an arc of capacity $u_{\max}$ and  $u$ parallel arcs of capacity $1$. Failure of the original arc corresponds to failure of the infinite capacity arc in the modified instance. See Figure~\ref{fig:two-capacities} for an illustration.
\end{proof}

In particular, this implies that our hardness result still holds in the more restricted setting of capacities $1$ and $\infty$ (where $\infty$ can also be replaced by a number that is polynomially bounded in the network size).

\begin{corollary}
  {\rflow} is NP-hard, even when restricted to instances where $u(e) \in \{1, \infty\}$ for all $e \in E$ and where the number of paths is polynomial in the size of the graph.
\end{corollary}

\paragraph{Unit capacities}
On the other hand, it is not hard to see that the problem becomes easy in the unit capacity case.

\begin{theorem}
If $u \equiv 1$ then any maximum flow also is a maximum robust flow.
\end{theorem}

\begin{proof}
  Let $C$ be a minimum $s$-$t$-cut in $G$. 
  If $|C| \leq k$, then every flow has robust value $0$. Thus assume $|C| > k$.
  Let $x$ be any $s$-$t$-flow. Clearly, $\robustval(x) \leq |C| - k$, as after removal of any $k$ arcs from $C$, the remaining flow must traverse the $|C| - k$ remaining arcs in the cut. Now assume $x$ is a maximum flow, i.e., $\flowval(x) = |C|$. Since every arc carries at most $1$ unit of flow, the removal of any $k$ arcs from $G$ can only decrease the flow value by $k$, thus $x$ is an optimal solution to {\rflow}.
\end{proof}

\section{Integral robust flows}\label{sec:integral}

In this section, we show that finding an maximum \emph{integral} robust flow is $\NP$-hard already for instances with $k = 2$. This is in contrast to the case $k = 1$, for which it is possible to compute the best integral solution.
In fact, our reduction implies that it is hard to distinguish instances with optimal value $2$ or~$3$, resulting in hardness of approximation for the integral problem.
Interestingly, the fractional version of the problem admits a $4/3$-approximation algorithm for~$k = 2$~\cite{bertsimas2013power}, indicating that the integral problem is indeed harder.

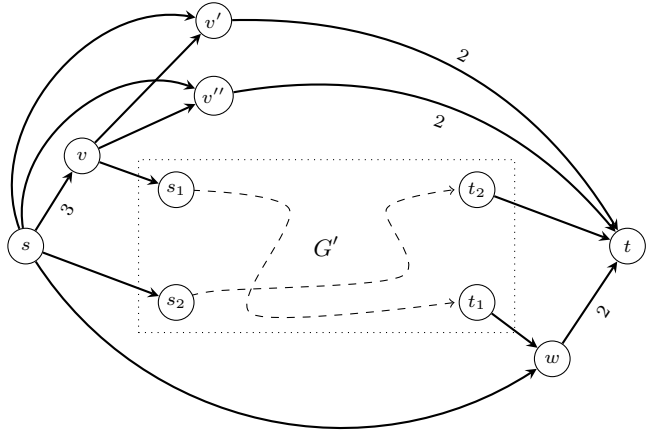
\begin{figure}[t]
  \centering
	\begin{tikzpicture}[font=\scriptsize]
		\node[labeledNode] (s) {$s$};
		
		\draw[dotted] (s) ++(1.5, 1.15) rectangle ++(5, -2.3);
		
		\path (s) ++(4, 0) node {\small $G'$};
				
		\draw[dashed, ->] plot [smooth] coordinates {(2, 0.75) (3.5, 0.5) (3, -0.9) (5.7, -0.75)};
				
		\draw[dashed, ->] plot [smooth] coordinates {(2, -0.75) (2.5, -0.6) (5, -0.4) (4.8, 0.5) (5.7, 0.75)};
		
		\path (s)
			++(0.75, 1.2) node[labeledNode] (v) {$v$}
				edge[normalEdge, <-] node[below, sloped] {$3$} (s)
			++(1.75, 1.8) node[labeledNode] (v1) {$v'$}
				edge[normalEdge, bend right=60, <-] (s)
				edge[normalEdge, <-] (v)
			++(0, -1) node[labeledNode] (v2) {$v''$}
				edge[normalEdge, bend right=60, <-] (s)
				edge[normalEdge, <-] (v);
				
		\path (s)
			++(2, 0.75) node[labeledNode] (s1) {$s_1$}
				edge[normalEdge, <-] (v)
		  +(4, 0) node[labeledNode] (t2) {$t_2$}
			++(0, -1.5) node[labeledNode] (s2) {$s_2$}
				edge[normalEdge, <-] (s)
			+(4, 0) node[labeledNode] (t1) {$t_1$};
			
		\path (s)
		  ++(8, 0) node[labeledNode] (t) {$t$}
				edge[normalEdge, <-] (t2)
				edge[normalEdge, bend right=30, <-] node[above, sloped] {$2$} (v1)
				edge[normalEdge, bend right=30, <-] node[below, sloped] {$2$} (v2)
			++(-1, -1.5) node[labeledNode] (w) {$w$}
				edge[normalEdge, bend left=45, <-] (s)
				edge[normalEdge, <-] (t1)
				edge[normalEdge, ->] node[below, sloped] {$2$}  (t);
			    
	\end{tikzpicture}
	\vspace{-0.7cm}
	\caption{The construction for showing $\NP$-hardness of {\irflow} with $k=2$. 
	The dotted box contains an instance of \textsc{Arc-disjoint Paths}.
	Labels at the arcs show the capacities. All unlabeled arcs have unit capacity.}
	\label{fig:integral-hardness}
\end{figure}

\begin{theorem}
  Unless $P = \NP$, there is no $(3/2 - \varepsilon)$-approximation algorithm for {\irflow}, even when restricted to instances where $k=2$ and $u(e) \leq 3$ for all $e \in E$.
\end{theorem}

\begin{proof}
  We reduce from \textsc{Arc-disjoint Paths}, which is well-known to be $\NP$-hard~\cite{FHW80}. 
  As input of \textsc{Arc-disjoint Paths}, we are given a directed graph $G' = (V', E')$ and two pairs of nodes $(s_1, t_1)$ and $(s_2, t_2)$. The task is to decide whether there is an $s_1$-$t_1$-path $P'_1$ and and $s_2$-$t_2$-path $P'_2$ in $G'$ with $P'_1 \cap P'_2 = \emptyset$.

  From the input graph $G' = (V', E')$, we construct an instance of {\rflow} by adding $6$ new nodes and $13$ new arcs, obtaining a new directed graph $G = (V, E)$ with
  \begin{align*}
    V & \ =\  V' \cup \{s,\, t,\, v,\, v',\, v'',\, w\}\\
    E & \ = \ E' \cup \{(s, v),\, (s, v'),\, (s, v''),\, (v, s_1),\, (v, v'),\, (v, v''),\,\\
    & \qquad\qquad\qquad (v', t),\, (v'', t),\, (s, w),\, (t_1, w),\, (w, t),\, (s, s_2),\, (t_2, t)\}.
  \end{align*}
  We set $u(s, v) = 3$ and $u(v', t) = u(v'', t) = u(w, t) = 2$. All other arcs have capacity $1$. The whole construction is depicted in \cref{fig:integral-hardness}.
  
  We show that there is an integral flow $x$ with $\robustval(x) \geq 3$ if an only if there is an $s_1$-$t_1$-path $P'_1$ and and $s_2$-$t_2$-path $P'_2$ in $G'$ with $P'_1 \cap P'_2 = \emptyset$. It is thus $\NP$-hard to distinguish instances of $\irflow$ with optimal value at least $3$ from those with optimal value at most $2$.
  
  First assume there is an integral flow $x$ with $\robustval(x) = 3$.
  Consider the arc set $S' = \{(s, v),\, (w, t)\}$.
  As $\sum_{P : P \cap S' = \emptyset} x(P) \geq 3$, there must be three $s$-$t$-paths  carrying $1$ unit of flow each and not intersecting with $S'$.
  These paths must thus start with the arcs $(s, v')$, $(s, v'')$, and $(s, s_2)$, respectively. 
  In particular, the latter path must end with $(t_2, t)$, as $t_2$ cannot be reached from $v'$ or $v''$. Let $P_2$ be this unique flow-carrying path starting with $(s, s_2)$ and ending with $(t_2, t)$. Note that all arcs of $P_2$ have unit capacity and thus no other flow-carrying path can intersect $P_2$.
  Now consider the arc set $S'' = \{(v', t),\, (v'', t)\}$. 
  Because the arc $(v, s_1)$ is part of an $s$-$t$-cut with capacity $3$ in the network $(V, E \setminus S'')$, there must be a flow path $P_1$ containing $(v, s_1)$. As $(t_2, t)$ is already saturated by the flow on $P_2$, the path $P_1$ must use $(t_1, w)$. In particular, $P_1$ contains an $s_1$-$t_1$-path $P'_1$ and $P_2$ contains an $s_2$-$t_2$-path $P'_2$, and $P_1 \cap P_2 = \emptyset$.
  
  Conversely, assume there is an $s_1$-$t_1$-path $P'_1$ and and $s_2$-$t_2$-path $P'_2$ in $G'$ with $P'_1 \cap P'_2 = \emptyset$. Let $P_1 = \{(s, v),\, (v, s_1)\} \cup P'_1 \cup \{(t_1, w),\, (w, t)\}$ and let $P_2 = \{(s, s_1)\} \cup P'_2 \cup \{(t_2, t)\}$. Send $1$ unit of flow along each of the paths $P_1, P_2$ and the five remaining paths $s$-$v'$-$t$, $s$-$v''$-$t$, $s$-$v$-$v'$-$t$, $s$-$v$-$v''$-$t$, and $s$-$w$-$t$, obtaining a flow $x$. Now assume by contradiction that $\robustval(x) < 3$. Because $\flowval(x) = 7$, there must be $S \in \mathcal{S}$ with $\sum_{P : P \cap S \neq \emptyset} x(P) > 4$. In particular, the arc $(s, v)$ must be contained in $S$, as it is the only arc carrying more than $2$ units of flow. The other arc in $S$ must be one of the arcs with capacity $2$, i.e., $(v, v')$, $(v, v'')$, or $(w, t)$. However, each of these three arcs is contained in one of the three flow paths using $(s, v)$. Thus $\sum_{P : P \cap S \neq \emptyset} x(P) = 4$, a contradiction.
\end{proof}

For the reduction given above to work, it is sufficient to have arcs of capacity at most~$3$. We now argue that the problem can be solved efficiently for arbitrary values of~$k$ when capacities are bounded by $2$.

\begin{theorem}
{\irflow} restricted to instances with $u(e) \leq 2$ for all $e \in E$ can be solved in polynomial time.
\end{theorem}

\begin{proof}
Let $x^*$ be an optimal solution to {\irflow}.
Let $x_1$ be a maximum flow in $G$ with respect to unit capacities and let $x_2$ be a maximum flow in $G$ with respect to capacities $u$. As capacities are integral, we can assume without loss of generality that $x_1$ and $x_2$ are integral. We will show that 
$\robustval(x^*) = \max \{0,\; \flowval(x_1) - k,\; \flowval(x_2) - 2k\}$.

To prove this claim, consider a minimum cardinality $s$-$t$-cut $C$ in $G$. 
We greedily construct a set $S \subseteq C$ with $|S| \leq k$ as follows: 
Starting with $S = \emptyset$, iteratively add an arc $e \in C \setminus S$ to $S$ that maximizes
$$\textstyle\Delta(S, e) := \sum_{P \in \mathcal{P} \st e \in P,\; P \cap S = \emptyset} x^*(P)$$
until $|S| = k$ or $S = C$.
In other words, we greedily add an arc from $C$ to $S$ that removes the most flow from $x^*$.
Note that throughout this selection process, the value $\Delta(S, e) \in \{0, 1, 2\}$ is non-increasing for each $e \in C$.
After construction of~$S$, let $\Delta := 0$ if $S = C$ and $\Delta := \max_{e \in C \setminus S} \Delta(S, e)$ otherwise.
\begin{enumerate}
\item If $\Delta = 0$, then $\sum_{P \in \mathcal{P} : P \cap S \neq \emptyset} x^*(P) = \flowval(x^*)$ and therefore $\robustval(x^*) = 0$. In this case, any integral flow is an optimal solution.
\item If $\Delta = 1$, then $\sum_{P \in \mathcal{P} \in \mathcal{P}: e \in P, P \cap S = \emptyset} x^*(P) \leq 1$ for every arc $e \in C \setminus S$. Therefore $\robustval(x^*) \leq \sum_{P : P \cap S = \emptyset} x^*(P) \leq |C \setminus S| = \flowval(x_1) - k \leq \robustval(x_1)$, where the equality follows from $\flowval(x_1) = |C|$ and the last inequality follows from the fact that every arc carries at most $1$ unit of flow in $x_1$. We conclude that $x_1$ is an optimal solution in this case.
\item If $\Delta = 2$, then $\sum_{P \in \mathcal{P} : P \cap S \neq \emptyset} x^*(P) = 2k$, because in every iteration an arc~$e$ with $\Delta(S, e) = 2$ was added to $S$. 
Thus $\robustval(x^*) \leq \flowval(x^*) - 2k \leq \flowval(x_2) - 2k \leq \robustval(x_2)$, where the last inequality follows from the fact that every arc carries at most $2$ units of flow in $x_2$. We conclude that $x_2$ is an optimal solution in this case.
\end{enumerate}
Computing~$x_1$ and~$x_2$ and identifying whether $\flowval(x_1) - k \geq \flowval(x_2) - 2k$ can be done in polynomial time.
\end{proof}

\section{Conclusion}
In this article, we point out that the computational complexity of the {\rflow} problem, which was believed to be settled 10 years ago, is indeed open. We show that the problem is $\NP$-hard when the number of failing arcs is part of the input. However, it remains a challenging open research question to determine the complexity when this number is bounded by a constant. Our hardness result also does not give bounds on the approximability of the problem (other than ruling out fully polynomial-time approximation schemes) and it would be interesting to see whether the $O(k)$-approximation by Bertsimas, Nasrabadi, and Orlin~\cite{bertsimas2013power} can be improved to a constant factor.

\paragraph{Acknowledgements}
We thank Tom McCormick and Gianpaolo Oriolo for numerous helpful discussions that led to the discovery of the flaw in \cite{DuChandrasekaran2007}. This work was supported by the `Excellence Initiative' of the German Federal and State Governments, the Graduate School~CE at TU~Darmstadt, and by the Alexander von Humboldt Foundation with funds of the German Federal Ministry of Education and Research (BMBF).

\printbibliography

\end{document}